\documentclass[conference,letterpaper]{IEEEtran}

\usepackage[utf8]{inputenc} 
\usepackage[T1]{fontenc}
\usepackage{url}
\usepackage{ifthen}
\usepackage{cite}
\usepackage{graphicx}
\usepackage{amssymb,amsfonts}
\usepackage{bbm}
\usepackage{mathtools}
\usepackage{amsthm}
\usepackage{amsmath} 
\usepackage{tikz}
\usepackage{pgfplots}
\usepgfplotslibrary{groupplots} 
\usetikzlibrary{pgfplots.groupplots}

\usepackage[numbers,sort&compress]{natbib}

\usepgfplotslibrary{fillbetween}

\newtheorem{theorem}{Theorem}

\newtheorem{lemma}{Lemma}
\newcommand{\Bernoulli}{{\rm Bernoulli}}

\newenvironment{lemmarep}[1]{\noindent {\bf Lemma #1.}\begin{it}}{\end{it}}

\usepackage{tikz,pgf,pgfplots}
\usetikzlibrary{calc,shadows}
\usetikzlibrary{pgfplots.groupplots}

\usepackage{times}
\usepackage{xr}

\usepackage{amsthm}
\usepackage{amsmath}    
\usepackage{amssymb}    
\usepackage{mathrsfs}   
\usepackage{epsfig}     
\usepackage{graphicx}
\usepackage{url}

\usepackage{url}
\usepackage{color}
\usepackage{booktabs}

\usepackage{subfigure}

\usepackage{enumitem}

\usepackage{pgfplotstable}
\usepackage{tikz,pgf,pgfplots}
\usetikzlibrary{calc,shadows}
\usetikzlibrary{pgfplots.groupplots}

\usepackage{ifthen}
\usepackage{mathtools}

\renewcommand{\caption}[1]{\singlespacing\hangcaption{#1}\normalspacing}

\usepackage{amsmath}
\usepackage{amsfonts}
\usepackage{dsfont}
\usepackage{mathrsfs}

\usepackage{algorithmic}
\usepackage{algorithm}
\usepackage{multirow}

\newif\iflong
\longtrue

\newif\ifdraft
\drafttrue

\newcommand{\M}{{\mathcal M}}

\newcounter{constcount}
\newcounter{numcount}
\newcounter{thmcnt}
  \let\Oldsection\section
  
\renewcommand{\section}{\stepcounter{thmcnt}\Oldsection}

\allowdisplaybreaks

\newcommand{\aln}[1]{\begin{align*}#1\end{align*}}

\newcommand{\al}[1]{\begin{align}#1\end{align}}

\renewcommand{\paragraph}[1]{
\vspace{2mm}
\noindent \textbf{#1}
}

\def\Item$#1${\item $\displaystyle#1$
   \hfill\refstepcounter{equation}(\theequation)}

\newcommand{\bea}{\begin{eqnarray}}
\newcommand{\eea}{\end{eqnarray}}
\newcommand{\beas}{\begin{eqnarray*}}
\newcommand{\eeas}{\end{eqnarray*}}

\newcommand\ML{ML} 

\newcommand\Tex{}
\newcommand\PR[2][\Tex]{
\ifthenelse{\equal{#1}{}}{{\mathrm{Pr}}\left(#2\right)}{\ensuremath{{\mathrm{Pr}}_{#1}\left[ #2\right]}}}

\newcommand\EX[2][\Tex]{
\ifthenelse{\equal{#1}{}}{{\mathbb E}\left[#2\right]}{\ensuremath{{\mathbb E}_{#1}\left[ #2\right]}}}

\newcommand\Var[2][\Tex]{
\ifthenelse{\equal{#1}{}}{{\mathrm{Var}}\left[#2\right]}{\ensuremath{{\mathrm{Var}}_{#1}\left[ #2\right]}}}
\renewcommand\M{M} 
\newcommand\len{L} 

\begin{document}

\title{Achieving the Capacity of a DNA Storage Channel with Linear Coding Schemes}


\author{%
  \IEEEauthorblockN{Kel Levick}
\IEEEauthorblockA{
University of Illinois, Urbana-Champaign\\
Urbana, IL, USA \\
klevick2@illinois.edu}
  \and
  \IEEEauthorblockN{Reinhard Heckel}
\IEEEauthorblockA{
Technical University of Munich \\
Munich, Germany \\
reinhard.heckel@tum.de}
  \and
\IEEEauthorblockN{Ilan Shomorony}
\IEEEauthorblockA{
University of Illinois, Urbana-Champaign\\
Urbana, IL, USA \\
ilans@illinois.edu}
}


\maketitle

\begin{abstract}
Due to the redundant nature of DNA synthesis and sequencing technologies, a basic model for a DNA storage system is a multi-draw ``shuffling-sampling'' channel. 
In this model, a random number of noisy copies of each sequence is observed at the channel output.
Recent works have characterized the capacity of such a DNA storage channel under different noise and sequencing models, relying on sophisticated typicality-based approaches for the achievability.
Here, we consider a multi-draw DNA storage channel 
in the setting of noise corruption by a binary erasure channel.
We show that, in this setting, the capacity is achieved by linear coding schemes.
This leads to a considerably simpler derivation of the capacity expression of a multi-draw DNA storage channel than existing results in the literature.
\end{abstract}

\begin{IEEEkeywords}
DNA storage, channel capacity, linear codes 
\end{IEEEkeywords}


\section{Introduction}

Due to its longevity and high information density, DNA has drawn growing interest in its potential for archival data storage. 
Thanks to recent advancements in DNA sequencing (reading) and synthesizing (writing), 
this idea is becoming practically viable, 
and several groups have recently demonstrated working DNA storage systems
\cite{dna1,dna2,dna3,dna4,dna5,dna6,dna7}. 
In these systems, data is usually stored on short DNA molecules (a few hundred nucleotides). 
The synthesis process is usually redundant and produces a large number of copies of each molecule.
At the time of reading, state-of-the-art sequencing technologies access this information, which corresponds to randomly sampling and reading sequences from the (redundant) DNA pool.
Additionally, sequencing and synthesis may introduce errors to each sequence, most commonly in the form of insertions, substitutions, and deletions. 

A natural mathematical model for DNA storage that accounts for these constraints is as follows: Data is stored onto $M$ sequences, each of length $L$. This can be thought of as a single codeword of length $ML$, broken into $M$ pieces of equal length. During sequencing, $N$ sequences are randomly drawn from this set.
Since the synthesis process is redundant and produces many copies of each molecule, and since the sequencing process is often preceded by Polymerase Chain Reaction (PCR), which effectively amplifies the number of copies of each molecule in the pool,
several of the sequenced molecules correspond to the same original input sequence.
However, because the synthesis and sequencing processes are noisy, the observed sequences are corrupted by \emph{distinct noise patterns}.
Therefore, an end-to-end model for DNA storage that captures this output sequence redundancy is the noisy shuffling-sampling channel with \emph{multi-draws}~\cite{lenz_ub}, shown in Figure~\ref{fig:multidraw}. 

First, each of the $M$ input strings are amplified a random 
number of times.
The resulting molecules are independently corrupted by a noisy channel and shuffled out of order.
Notice that some molecules may be sampled zero times, corresponding to the case where they are not sequenced at all.
The decoder must then, without any knowledge of which molecules were sampled, reconstruct the original stored message using the sampled sequences at the channel output. 

\begin{figure}[!t]
    \centering
    \includegraphics[width=0.9\linewidth]{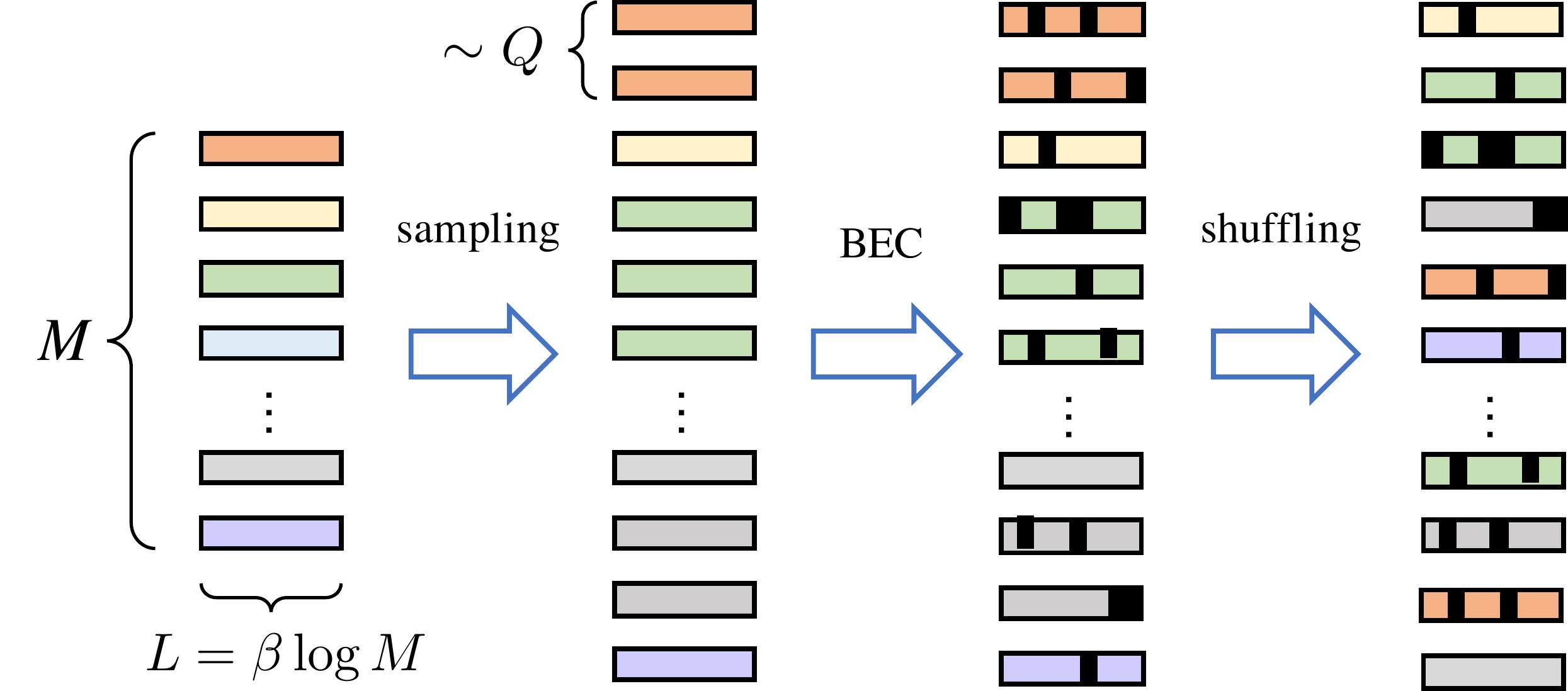}
    \caption{The BEC multi-draw DNA Storage channel.}
    \label{fig:multidraw}
\end{figure}

In this multi-draw setting, a clustering problem arises from the need to identify which of the output sequences correspond to the same input sequence.
If we are able to correctly cluster the output, we can combine the sequences in each cluster to correct errors and decode the stored message using the (partially) error-corrected sequences.


The capacity of the multi-draw DNA storage channel with binary symmetric noise has been established using typicality based arguments~\cite{lenz_ub,lenz_ach} and for general discrete memoryless channels based on the method of types~\cite{weinberger2021dna}. 
While these arguments are interesting and novel, the analysis of the achievable rate is sophisticated and does not provide direct intuition for the resulting capacity expression.
Moreover, the resulting decoding algorithms are computationally intractable,  raising the question of whether simpler approaches such as linear codes can achieve the channel capacity.

Motivated by the fact that the capacity of a binary erasure channel (BEC) can be straightforwardly achieved using linear codes, in this work we consider the multi-draw DNA storage channel with binary erasure noise.
More precisely, we focus on a multi-draw shuffling-sampling channel where each (binary) sequence is corrupted by a binary erasure channel with erasure probability $p$, and each input string is drawn a random $Q$ number of times, where $Q$ has the probability mass function $\text{Pr}(Q=n)=q_n$, for $n=0,1,2,\dots$.
As in previous works, we consider the asymptotic regime where $M \to \infty$ and the read length scales as $L = \beta \log M$.
The simplicity of the BEC setting allows us to 
show that a (random) linear coding scheme achieves the capacity of this channel for a large set of parameters $(p,\beta)$,
illustrated in Figure~\ref{fig:regimes0}.

Based on the proposed linear coding scheme, we show that, for the blue regime in Figure~\ref{fig:regimes0}, the capacity is given by
\al{
    C =(1-q_0)(E[C_Q|Q\geq 1]-1/\beta),
\label{eq:generalcap}
}
where $C_n$ is the capacity of a multi-draw shuffling-sampling channel with exactly $n$ draws of each input sequence. 
As it turns out, the capacity expression in (\ref{eq:generalcap}) can be verified to be equivalent to the expression obtained in related works \cite{lenz_ub,weinberger2021dna}.
However, the expression in (\ref{eq:generalcap}) is more intuitive and directly recovers 
all previous DNA storage channel capacity results, including the original results for DNA storage channels with ``one or none draws'' ($q_0 + q_1 = 1$) \cite{dna_sh}, in which case $E[C_Q | Q \geq 1] = C_1$.
Hence, we conjecture that (\ref{eq:generalcap}) is the general capacity formula for an arbitrary multi-draw DNA storage channel.

\subsection{Related literature}\label{related_literature}


The information-theoretic analysis of DNA storage channels started in \cite{heckel2017fundamental} with a noise-free shuffling-sampling channel, and was later extended to noisy shuffling-sampling channels~\cite{shomorony2019capacity,dna_sh} by modeling the noise as a BSC,
and considering a single-draw setting
where strings are drawn either once with probability $1-q_0$ or not at all with probability $q_0$.

The single-draw DNA storage channel with BEC noise was considered 
in \cite{shin}.
The concept of 
\textit{consistency}, where two strings $x_1^L$, $x_2^L$ with erasures are said to be consistent if they agree on every non-erased position, 
is used to establish the capacity for a set of parameters $(p,\beta)$.
We will also make use of the notion of consistency to create consistency graphs from the output strings
in Section~\ref{ach}.

The multi-draw shuffling-sampling channel was first studied in \cite{lenz_ub,lenz_ach}. 
The capacity was characterized for a regime of $(p,\beta)$ and for the case where the output strings are independently observed through a BSC.
The achievability argument was based on a random codebook construction.
The decoder performs a greedy-like clustering of the output strings, and then uses typicality decoding based on a new notion of typicality between a set of $d$ output strings and an input string. 

Capacity results for multi-draw DNA storage channels were recently generalized to arbitrary discrete memoryless channels~\cite{weinberger2021dna}.
A general achievability was provided based on the method of types and a general upper bound was developed by refining the approach used in previous works~\cite{dna_sh,lenz_ub}.


The problem of clustering output strings was studied from a coding-theoretic standpoint in~\cite{clusteringcorrecting}.
It was shown that ``code-aware'' clustering, i.e., a clustering algorithm that exploits knowledge of the codebook (as opposed to a code-oblivious clustering), can reduce the number of redundancy bits needed for correct decoding.

\begin{figure}
\centering 
\vspace{4mm}
\begin{tikzpicture}\label{fig:regimes0}

    \begin{semilogxaxis}[xlabel=$p$,
        ylabel=$\beta$,
        ylabel near ticks,
        xmin=0.001,
        xmax=1,
        ymin=0,
        ymax=10,
        width=0.8\linewidth,
        axis on top]

   \addplot[name path=redrect, fill=red, fill opacity=0.5, draw=none, mark=none]
coordinates {
    (0.00001, 0)
    (0.00001, 1)
    (1, 1)
    (1, 0)
};

\addplot[name path=fout,color=green!50!black,line width=2pt,mark=none] table[x index=2,y index=3]{./curves.dat};   

\path[name path=axis] (axis cs:0.00001,20) -- (axis cs:0.1,20);
\addplot[green!80!white!100, fill opacity = 0.4] fill between[of=fout and axis];

\addplot[name path=f,color=blue!80!black,line width=2pt,mark=none] table[x index=0,y index=1]{./newcurves.dat};   
   
\path[name path=axis] (axis cs:0.00001,20) -- (axis cs:0.1,20);
\addplot[blue!20] fill between[of=f and axis];

\path[name path=axis2] (axis cs:0.00001,0.9) -- (axis cs:1,1) -- (axis cs:1,20) -- (axis cs:0.1,20);
\addplot[gray!20] fill between[of=fout and axis2];

\end{semilogxaxis}
\end{tikzpicture}  
\vspace{2mm}
\caption{\label{fig:regimes0}
The outer bound from (\ref{eq:multiconverse}) holds in the blue region, which corresponds to $\beta > 2/(1-2p+p^2)$.
The inner bound from Theorem~\ref{th:multiachievability} holds above the green line, which corresponds to $\beta > -1/\log(1-\tfrac12(1-p)^2)$. 
This characterizes the capacity of the BEC shuffling channel in the blue region as $(1-q_0)(1 - p_{\text{eff}} - 1/\beta)$.
The capacity in the red region (i.e., for $\beta < 1$) is $0$ and it is unknown in the gray region.}
\end{figure}
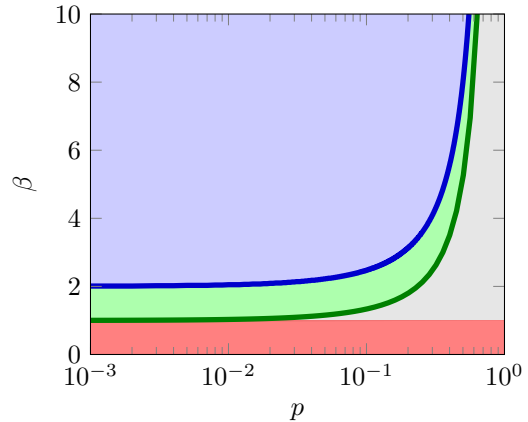

\section{Main Result}\label{main_result}

We consider the multi-draw DNA storage channel illustrated in  Figure~\ref{fig:multidraw}.
The channel input is a length-$ML$ binary string
$
X^{\ML} = \left[ X_1^\len, X_2^\len, \ldots , X_M^\len \right],
$
or, equivalently, $M$ strings of length $L$ concatenated to form a single string of length $ML$.
Each of the $M$ input strings is independently sampled $Q$ times, where $\Pr(Q  = n) = q_n$, for $n=0,1,\dots$, yielding a set of $N$  sequences.
Each of these sequences is independently passed through an binary erasure channel with erasure probability $p$, and the final $N$ sequences are shuffled out of order.
We refer to the resulting end-to-end channel as the BEC multi-draw DNA storage channel.

As in previous works, we let $\len = \beta \log \M$, and consider the asymptotic regime $\M \to \infty$.
A rate $R$ is said to be achievable if there exists a sequence of codes, each with $2^{\M \len R}$ codewords, and whose error probability goes to zero as $\M \to \infty$.
The channel capacity $C$ is the supremum over all achievable rates.

We say that a code is a \emph{linear coding scheme} if the $2^{ML R}$ length-$ML$ codewords are all in the range of a $ML \times B$ binary generator matrix ${\bf G}$.
Notice that we differentiate this from a \emph{linear code}, in which case the set of codewords must be the entire range of ${\bf G}$.
Our main result is the following.




\begin{theorem}\label{th:main_result}
For $\beta>2/(1-2p+p^2)$,
the capacity of the BEC multi-draw DNA storage channel is
\begin{align}\label{eq:multicapacity}
C = (1-q_0) \left( E[C_{\text{BEC},Q}|Q\geq 1] - 1/\beta \right),
\end{align}
and it can be achieved with a linear coding scheme.
Here, $C_{\text{BEC},n}=1-p^n$ is the capacity of a BEC with $n$ draws.
\end{theorem}


A natural approach to deal with the output of a BEC multi-draw DNA storage channel is to  first cluster output sequences based on consistency, and then combine all of the strings in each cluster into a ``consensus" sequence, 
where the $i$th position of the length-$L$ consensus sequence is the $i$th symbol from any of the strings of the cluster that is not erased. 

If the clustering can be done successfully, 
this reduces the probability of erasure to $p^n$ for an output cluster with $n$ strings. 
Therefore, we expect that after consensus there will be $(1-q_0)M$ output strings, and for $n=1,2,\dots$, there will be $q_nM$  consensus strings with erasure probability $p^n$. 
The resulting effective channel after the consensus step is similar to the single-draw BEC case as discussed in \cite{shin}, but rather than a bit erasure probability of $p$ across all strings, here the average effective erasure probability is given as
\begin{equation}\label{eq:peff}
    p_{\text{eff}}\triangleq E[p^{Q}|Q\geq 1]=\frac{\sum_{n=1}^\infty q_np^n}{1-q_0}.
\end{equation}
Note that $E[C_{\text{BEC},Q}|Q\geq 1]= 1-E[p^{Q}|Q\geq 1]=1-p_{\text{eff}}$.
Since the capacity of BEC single-draw DNA storage channel \cite{shin} can be found to be 
\aln{
(1-q_0) (1-p - 1/\beta),
}
it is natural to conjecture that the capacity of the BEC multi-draw DNA storage channel is given by \eqref{eq:multicapacity}.


The converse to Theorem~\ref{th:main_result} can be found by considering a genie-aided argument where the genie reveals the true clusters, which can be used to find consensus sequences, and then using the result for the single-draw setting~\cite{shin}.
More formally, the converse can be obtained from the general DMC DNA storage channel capacity given in \cite[Corollary~11]{weinberger2021dna}. 
Evaluating this result for the BEC yields
\begin{equation}\label{eq:multiconverse}
    C = (1-q_0)(1-p_{\text{eff}}-1/\beta)
\end{equation}
for the regime $\beta > 2/(1-2p+p^2)$. This is represented by the blue area in Figure \ref{fig:regimes0}. 
Next, we show that for a larger set of parameters $(p,\beta)$ (given by the green region), the capacity expression in \eqref{eq:multicapacity} can be achieved with linear coding schemes.


\section{Achievability via Linear Schemes}\label{ach}

Motivated by
the fact that  linear codes achieve the capacity of the BEC \cite{elias}, we show that linear coding schemes achieve the capacity of the BEC multi-draw DNA storage channel. 

To construct a code of rate $R$, we first populate a random binary generator matrix $\mathbf{G}$ of size $ML\times B$, for some $B$ to be determined, with i.i.d. Bernoulli(1/2) entries. We then generate $2^{MLR}$ length-$B$ random binary vectors $\mathbf{t}_i$, also with i.i.d. Bernoulli(1/2) entries. The $i$th codeword of our random code, for $i=1,\dots,2^{MLR}$, is constructed first by computing the product $\mathbf{Gt}_i$ over $\mathbb{F}_2$, and then by breaking the resulting codeword into $M$ binary strings of length $L$. The channel input of $M$ strings thus represents a single codeword. 

We will use the 
following lemma
throughout the proofs in this section.


\begin{lemma}\label{lm:full_rank}
Let $\bf G$ be an $ML \times B$ matrix with  i.i.d.~$\Bernoulli(1/2)$ entries.
Fix any $\delta \in  (0,1)$ and a submatrix $\bf G'$ formed by an arbitrary set of $(1-\delta)B$ rows of $\textbf{G}$.
Then $\bf G'$ is full rank (over the finite field $\mathbb{F}_2$) with probability tending to $1$ as $B \to \infty$.
\end{lemma}

\begin{figure}[!t]
    \centering
    \includegraphics[width=0.9\linewidth]{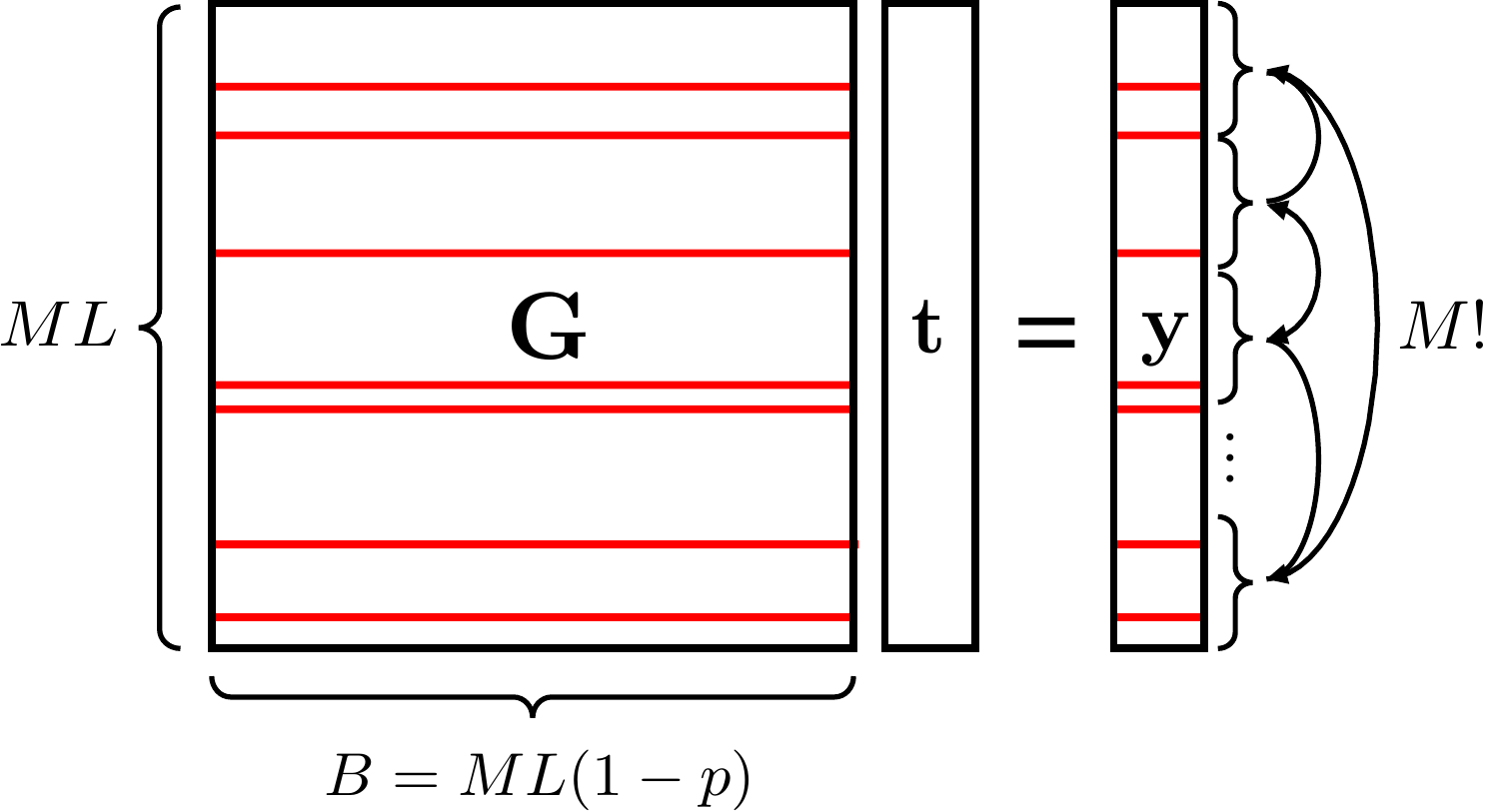}
    \caption{In the setting where each sequence is observed exactly once at the output, there are $M!$ potential systems of equations, each with $ML(1-p)$ non-erased equations. Choosing our rate to be $\approx 1-p-1/\beta$ guarantees that, with high probability, only one of these systems will have a solution that corresponds to a codeword.}
    \label{fig:system}
\end{figure}

\subsection{Single-draw case} 
\label{sec:singledraw}

We first illustrate the linear coding scheme by considering the case where each string is sampled exactly once, i.e., $q_1=1$. At the output of this system, $N=M$ strings are observed, which we wish to use to recover the stored message $\mathbf{t}_i$.

If the correct ordering of the $M$ output strings were known, we would be able to concatenate them into a length-$ML$ vector $\mathbf{y}$ and try to solve the system $\mathbf{Gt}=\mathbf{y}$ for $\mathbf{t}$. With erasure probability $p$, in expectation there are $ML(1-p)$ non-erased positions in $\mathbf{y}$, so a unique solution to this system exists with high probability as long as the number of remaining equations, which is roughly $ML(1-p)$, is greater than or equal to the number of variables $B$. Therefore, $B$ can be set to $ML(1-p-\epsilon)$ for any $\epsilon>0$ and all binary strings in $\{0,1\}^B$ may be used as message vectors.

However, the correct ordering of the output strings is not actually known, so we consider all $M!$ possible orderings. With each ordering, we have a distinct concatenated vector $\mathbf{y}$ with a $p$ fraction of erasures, as well as a $p$ fraction of useless equations in $\mathbf{G}$, as shown in Figure \ref{fig:system}. From Lemma \ref{lm:full_rank}, if we set $B=ML(1-p-\epsilon)$, then the true ordering of remaining equations will have a unique solution.

In addition, we must have only one feasible ordering of the output strings; that is, out of the $M!$ possible systems of equations $\mathbf{Gt}=\mathbf{y}$, only one of them (the true one) should have a solution $\mathbf{t}$ that is one of the original message vectors $\mathbf{t}_i$, for $i=1,\dots,2^{MLR}$. 
Assume without loss of generality that  message 1 is sent (i.e., the codeword sent is $\mathbf{Gt}_1$). Since the messages are generated i.i.d. from $\{0,1\}^B$, by the union bound, the probability that there is a collision between the solution $\mathbf{t}$ of one of the $M!-1$ incorrect systems and one of the other messages $\mathbf{t}_i$, $i=2,\dots,2^{MLR}$ is at most
\begin{align*}
(M! - 1) 2^{MLR} 2^{-B} &< 2^{M \log M + MLR - B}\\
&= 2^{ML[ R - (1-p-\epsilon  -1/\beta)]}.    
\end{align*}
Therefore, by choosing $\epsilon$ arbitrarily small, we conclude that any rate $R<1-p-1/\beta$ can be achieved with vanishingly small error probability.

\begin{figure}[!t]
    \centering
    \includegraphics[width=0.8\linewidth]{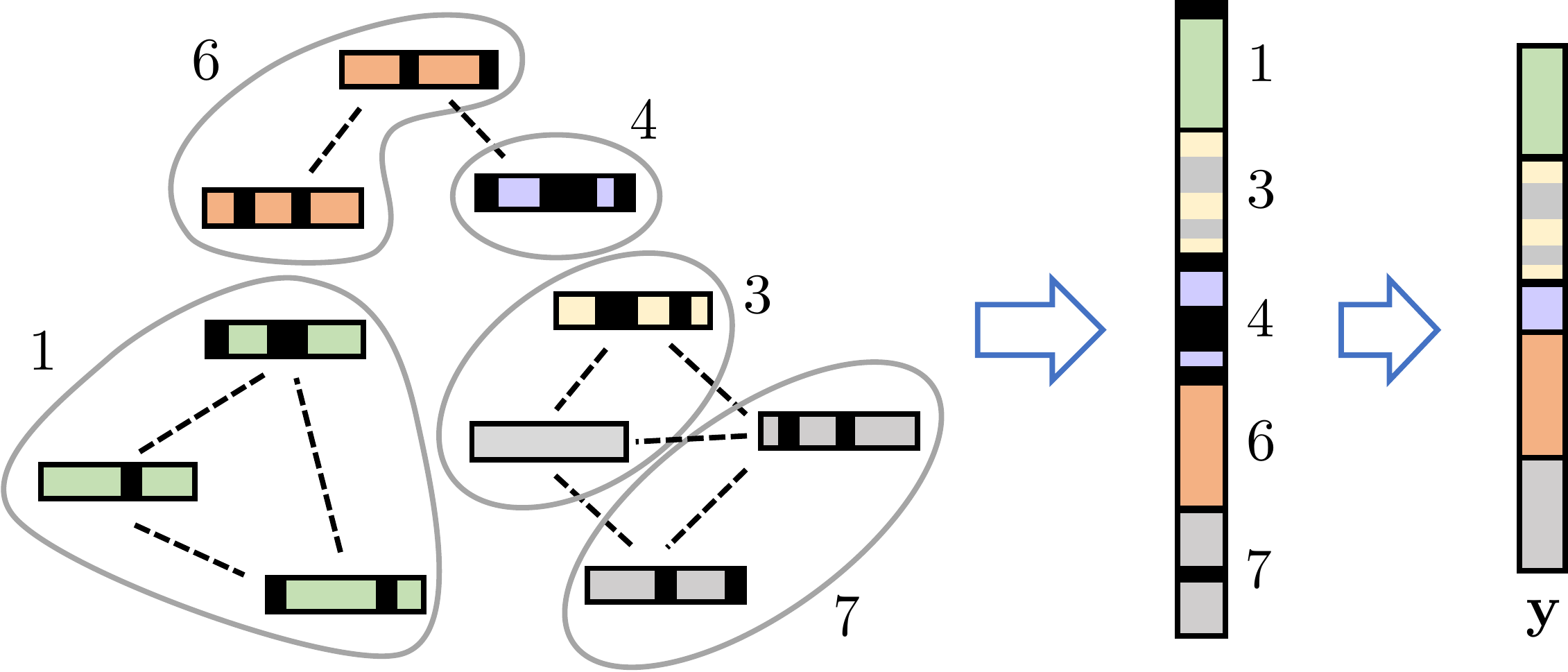}
    \caption{The decoder first builds a consistency graph between all received sequences, and considers all possible partitions of the graph into cliques, where the total number of  cliques is between $p_LM$ and $p_UM$.
     For each such valid clustering, the decoder considers all possible assignments of the indices $\{1,...,M\}$ to the clusters and uses those indices to order the consensus sequences of each cluster to form the vector 
    $\mathbf{y}$.
    After removing erased entries and the rows of $\mathbf{G}$ corresponding to erased/missing rows in $\mathbf{y}$, the system $\mathbf{G} \mathbf{t} = \mathbf{y}$ can be solved.}
    \label{fig:clustering-decoding}
\end{figure}

\subsection{Multi-draw case}

We now consider the channel in Figure \ref{fig:multidraw} with a general sampling distribution $Q$, i.e., each input string $x_i^L$ is sampled $N_i\sim Q$ times. We now expect to see $E[N_1]M$ output strings, and we would like to cluster them into roughly $(1-q_0)M$ clusters. This clustering task becomes easier under the BEC setting, as we can take advantage of the consistency of the strings. Notice that any two output strings that originated from the same input string are consistent. Therefore, given $N$ output strings $x_1^L,\dots,x_N^L$, one can construct an undirected graph with the strings as vertices and edges between any two consistent strings. We refer to this graph as a \textit{consistency graph}. A clustering of these output strings is valid if it corresponds to a partition of $\{x_1^L,\dots,x_N^L\}$ such that each group corresponds to a \textit{clique} in the consistency graph. 

With the probability of an input string not having an output cluster equal to $q_0$, in expectation there are $(1-q_0)M$ output clusters. From Hoeffding's inequality, the probability that there are more than $p_UM\coloneqq(1-q_0+\epsilon)M$ or fewer than $p_LM\coloneqq(1-q_0-\epsilon)M$ output clusters can be bounded as 
\begin{align*}
    \Pr\left( \left| \text{\# true output clusters} - (1-q_0)M \right| > \epsilon M \right) < 2e^{-2M\epsilon^2},
\end{align*}
which tends to 0 as $M\to\infty$ for any $\epsilon>0$. The task of the decoder is then to cluster the $N$ output strings into between $p_LM$ and $p_UM$ clusters for some small fixed $\epsilon$; here, the decoder will consider all valid clusterings that create between $p_LM$ and $p_UM$ clusters. For each clustering, the decoder first performs a consensus step, effectively converting the $N$ clustered output strings into $(1-q_0)M$ strings with a smaller erasure rate. From here, the decoder proceeds similarly to the single-draw case. Each cluster is assigned a distinct label from $\{1,\dots,M\}$, and the clusters are ordered by label to create a single output vector $\mathbf{y}$ of length roughly $(1-q_0)ML$. All of the at most $M!$ possible label assignments are considered. The system of equations corresponding to each label assignment and vector $\mathbf{y}$ can be solved for a solution $\mathbf{t}$, if a solution exists. This cluster-based decoding scheme is illustrated in Figure \ref{fig:clustering-decoding}. 

Again, we must choose $B$ large enough so that the true system, obtained by clustering and ordering the output strings correctly, has a unique solution. Additionally, we must have $R$ small enough so that only one of the systems (the true one) yields a solution that corresponds with a valid message vector $\mathbf{t}_i$ so that the correct message can be decoded. 

To guarantee a unique solution, the true system must have enough equations after erasures have been discarded. 
After the consensus step has been performed on the output clusters, we have at least $p_LM$ output strings and an expected effective erasure probability of $p_{\text{eff}}$. 
It can be shown using standard concentration inequalities that the effective erasure probability cannot deviate significantly from $p_{\text{eff}}$. 
Thus, the true system will have at least $MLp_L(1-p_{\text{eff}}-\epsilon)$ equations with high probability, and if we set $$B=MLp_L(1-p_{\text{eff}}-\epsilon)(1-\epsilon),$$
then by Lemma \ref{lm:full_rank}, the true system of equations will have a unique solution with probability tending to 1 as $M\to\infty$. 

To find the maximum rate $R$ for which this scheme succeeds with vanishing error probability, we must bound the total number of valid clusterings of the output strings. 
We begin by analyzing the total number of edges in the consistency graph. 
Since each cluster of size $n$ produces $\binom{n}{2}\leq n^2/2$ ``correct'' edges (i.e., edges between output strings that originated from the same input string), the expected number of correct edges is at most 
\begin{align} \label{eq:correct}
\sum_{n=0}^\infty M q_n \frac{n^2}{2} = \frac{M}{2} E[Q^2].
\end{align}

Now let $Z$ be the total number of incorrect edges in the consistency graph, and define 
\begin{align*}
    \gamma \coloneqq - \beta \log \left( 1-\tfrac12(1-p)^2 \right),
\end{align*}
which is positive for any $p\in(0,1)$. 
\begin{lemma}
\label{lm:edges}
The number of incorrect edges $Z$ satisfies
\begin{align}
\Pr\left(Z > M^{2-\gamma + \epsilon}\right) \to 0\end{align}
as $M \to \infty$, for any $\epsilon > 0$. 
\end{lemma}

From Lemma \ref{lm:edges}, we can see that as long as $\gamma > 1$, the number of incorrect edges grows slower than $M$, and will therefore be vanishingly small compared to the number of correct edges given in \eqref{eq:correct}. 

\begin{figure}[!t]
    \centering
    \includegraphics[width=0.9\linewidth]{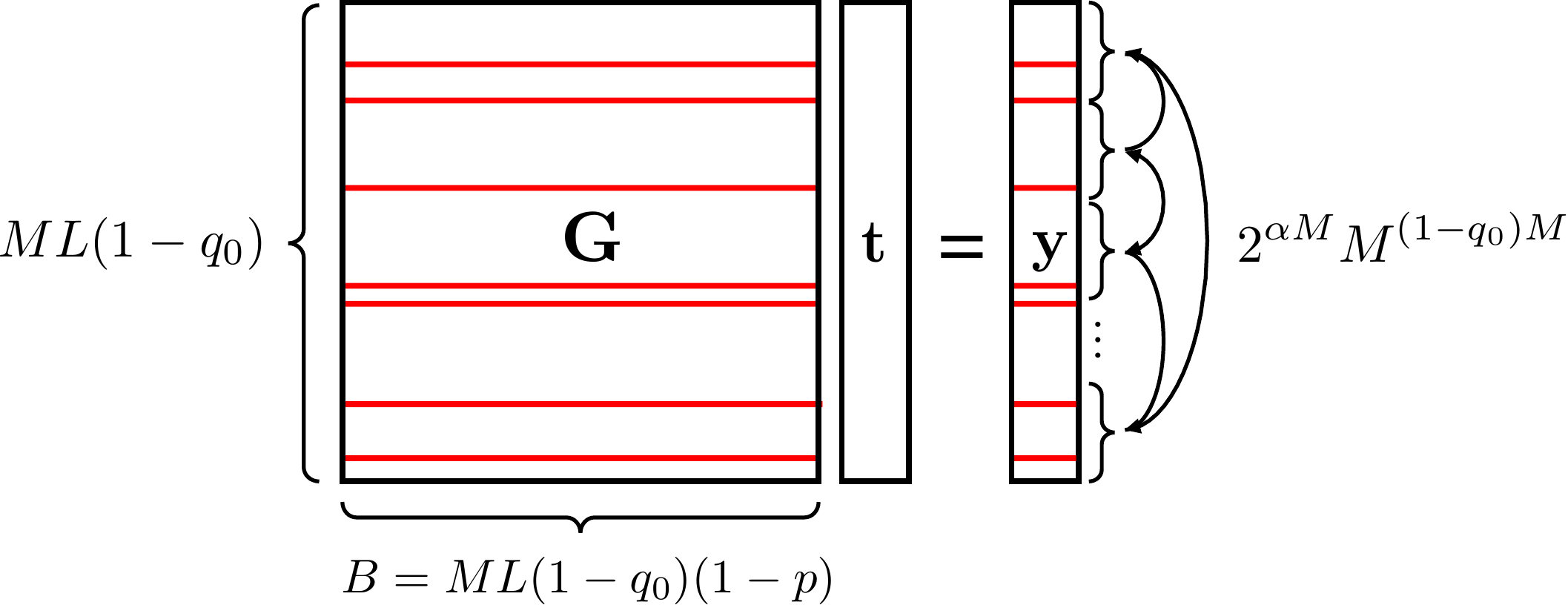}
    \caption{In the multi-draw setting, there are $2^{\alpha M} M^{(1-q_0)M}$ potential systems of equations, obtained by considering all valid ways to cluster the output strings into $(1-q_0)M$ clusters, and then assigning each cluster to an index.
     The true system (corresponding to the correct clustering and ordering) has $ML(1-q_0)(1-p_{\text{eff}})$ non-erased equations in expectation.}
    \label{fig:system_multi}
\end{figure}

We now bound the number of possible matchings given the total number of edges and find that a loose bound is sufficient to establish capacity.

\begin{lemma} \label{lm:clusterings}
Suppose the consistency graph has a total of $U$ edges.
Then there are at most $2^U$ valid ways to cluster the output sequences.
\end{lemma}

\begin{proof}
Each valid clustering corresponds to a partition of the output sequences such that each group corresponds to a clique in the consistency graph.
Notice that a partition of the consistency graph into cliques is uniquely described by the set of edges that are part of each of the cliques.
Hence, each partition corresponds to a distinct subset of the $U$ edges in the graph.
Since there are at most $2^U$ such subsets, it follows that there are at most $2^U$ valid ways to cluster the output sequences.
\end{proof}

Following \eqref{eq:correct} and Lemma \ref{lm:edges}, we know that for some $\gamma > 1$, the number of edges $U$ in the consistency graph satisfies $U<\alpha M$ with high probability, for some $\alpha>1$. Thus, Lemma \ref{lm:clusterings} implies that the number of ways to cluster the output sequences is at most $2^{\alpha M}$. 

Now, similar to the single-draw achievability proof (Section~\ref{sec:singledraw}), we must guarantee that there is no collision between the solution found $\mathbf{t}$ and any incorrect codeword $\mathbf{t}_i$, for $i=2,\dots,2^{MLR}$ (assuming message 1 is sent). The decoder must attempt to solve a total of at most $2^{\alpha M}M^{p_UM}$
systems of equations, as we need to cluster the output strings into at most $p_UM$ valid clusters and assign each cluster an index in $\{1,\dots,M\}$ for the ordering. 
The probability of a collision is then upper-bounded by
\begin{align}
    \label{eq:errormultidraw}
2^{\alpha M} M^{p_UM} 2^{MLR} 2^{-B} & = 2^{\alpha M + p_U M \log M + MLR - B} \nonumber \\
& = 2^{ML( \alpha/L + p_U/\beta + R - B/(ML))},
\end{align}
which goes to 0 as $M\to\infty$ as long as
\begin{align*}
    R + \frac{\alpha}{\beta \log M} - p_L(1-p_{\text{eff}} -\epsilon)(1-\epsilon) + \frac{p_U}{\beta} < 0.
\end{align*}
Since $\alpha/(\beta\log M)\to 0$ and $\epsilon$ can be chosen arbitrarily small, 
\begin{align*}
    R < (1-q_0)(1-p_{\text{eff}} - 1/\beta)
\end{align*}
is achievable, and we have
the following capacity lower bound:

\begin{theorem} \label{th:multiachievability}
The capacity of the BEC multi-draw DNA storage channel satisfies
\begin{align}
    C_{\text{BEC, multi-draw}} \geq (1-q_0) \left( 1 - p_{\text{eff}} - 1/\beta \right),
\end{align}
as long as $\gamma = - \beta \log \left( 1-\tfrac12(1-p)^2 \right) > 1$.
\end{theorem}

Since the parameter regime required for the upper bound in Equation \ref{eq:multiconverse} is smaller than that required for the lower bound, we have that Theorem \ref{th:main_result} holds for the smaller regime $\beta > 2/(1-2p+p^2)$ 
(the blue region in Figure \ref{fig:regimes0}).

\section{Discussion}
A natural follow-up to this work is to investigate whether the random linear coding scheme here presented also achieves the capacity of the BSC multi-draw DNA storage channel.
This is reasonable, as linear codes are also known to achieve the capacity of a BSC~\cite{elias}.
However, this problem is more complicated than the case of the BEC, as the concept of consistency does not exist with bit substitution errors. 

Note that the clustering algorithm covered here is ``code-aware''; the decoder considers every feasible clustering of the output strings and only permits a clustering whose corresponding linear equations have a unique, valid solution (message index). 
Hence, another natural follow-up question is whether a code-oblivious clustering algorithm is also sufficient to achieve capacity.
Code-oblivious clustering approaches are more desirable as they provide a separation between the clustering and decoding tasks.
Notice that the greedy clustering algorithm used in \cite{lenz_ach} for the BSC case is code-oblivious. 


Finally, we remark that the capacity for a large set of parameters $(p,\beta)$ is still an open question. 
Weinberger and Merhav \cite{weinberger2021dna} provide the capacity for a general DMC,
which in the BEC case corresponds to the blue region in Figure~\ref{fig:regimes0}.
Characterizing the rates achieved by linear coding schemes and code-aware clustering in the gray region in Figure~\ref{fig:regimes0} is another direction for future work.



\appendices

\section{Proof of Lemma \ref{lm:full_rank}}

\begin{lemmarep}{\ref{lm:full_rank}}
Let $\mathbf{G}$ be an $ML \times B$ matrix with  i.i.d.~$\Bernoulli(1/2)$ entries.
Fix any $\delta \in  (0,1)$ and a submatrix $\mathbf{G}'$ formed by an arbitrary set of $(1-\delta)B$ rows of $\mathbf{G}$.
Then $\mathbf{G}'$ is full rank (over the finite field $\mathbb{F}_2$) with probability tending to $1$ as $B \to \infty$.
\end{lemmarep}

\begin{proof}
We follow the approach in the lecture notes by \cite{sayir_notes}.
In order for $\mathbf{G}'$ to be full rank, the $(n+1)$th row must be chosen as a vector that is not in the span of rows $1,\dots,n$.
Note that the space spanned by $n$ linearly independent vectors in $\mathbb{F}_2$ has  exactly $2^n$ distinct elements.
If we assume that the first $n$ rows are linearly independent, then the probability that the $(n+1)$th row (which is a $B$-dimensional vector) is not in the span of the first $n$ rows is $1- 2^{n-B}$.
By induction we see that the probability that all $(1-\delta)B$ rows are linearly independent is 
\begin{align}
    \prod_{j=1}^{(1-\delta)B} \left(1-2^{-(B-j+1)}\right) 
& = \prod_{i=\delta B + 1}^{B} \left(1-2^{-i}\right) \nonumber \\
& = \frac{\prod_{i=1}^{B} \left(1-2^{-i}\right)}{\prod_{i=1}^{\delta B} \left(1-2^{-i}\right) }. \label{eq:lemlimit}
\end{align}
As $B \to \infty$, both the product in the numerator and in the denominator can be verified (e.g., using numerical software) to converge to
\begin{align*}
    \prod_{i=1}^{\infty} \left(1-2^{-i}\right) = 0.28879...    
\end{align*}
which implies that (\ref{eq:lemlimit}) tends to $1$ as $B \to \infty$, proving the lemma.
Notice that, if $\delta = 0$, the probability does not tend to $1$ and instead tends to $0.28879$, which is in contrast to the case of a real-valued matrix with random entries from a continuous distribution, where all square submatrices will be full-rank with high probability.
\end{proof}

\section{Proof of Lemma \ref{lm:edges}}

\begin{lemmarep}{\ref{lm:edges}}
The number of incorrect edges $Z$ satisfies
\begin{align}
    \Pr\left(Z > M^{2-\gamma + \epsilon}\right) \to 0
\end{align}
as $M \to \infty$, for any $\epsilon > 0$. 
\end{lemmarep} 

\begin{proof}
Consider two output strings $y^L_i$ and $y^L_j$ that are generated from distinct input strings $x^L_i$ and $x^L_j$.
We show that $y^L_i$ and $y^L_j$ are consistent 
with probability $M^{-\gamma}$.

Let $x^L_i[\ell]$ and $x^L_j[\ell]$, for $\ell=1,\ldots,L$ be the individual symbols in the sequences.
Notice that $x^L_i$ and $x^L_j$ are generated by choosing $\mathbf{t}_i$ and $\mathbf{t}_j$ uniformly at random from $\{0,1\}^B$, and computing $\mathbf{G}'\, \mathbf{t}_i$ and  $\mathbf{G}''\, \mathbf{t}_j$, where $\mathbf{G}'$ and $\mathbf{G}''$ are each obtained by taking the $L$ rows from $\mathbf{G}$ corresponding to the $i$th and $j$th input sequences (note that $\mathbf{G} \mathbf{t}_i$ has length $ML$).

First we claim that the $2L$ random variables $x^L_i[\ell],x^L_j[\ell]$, $\ell=1,\ldots,L$, are mutually independent Bernoulli$(1/2)$.
Treating all vectors as column vectors, we can write
\begin{align}\label{eq:xgt}
    \begin{bmatrix}
x^L_i \\ x^L_j
\end{bmatrix} = 
\underbrace{
\begin{bmatrix}
\mathbf{G}' & 0 \\ 0 & \mathbf{G}''
\end{bmatrix}}_{H}
\underbrace{
\begin{bmatrix}
\mathbf{t}_i \\ \mathbf{t}_j
\end{bmatrix}}_{\tilde{\bf t}}.
\end{align}
The block diagonal matrix $H$ above has dimension $2L \times 2B$, where $B = ML (1-q_0 -\epsilon)(1-p_{\text{eff}}-\epsilon)(1-\epsilon)$.
For $M$ large enough, we have $B > L$, and $H$ is full row-rank, with a null space of dimension $2B - 2L$.
Hence, for any 
$\mathbf{c} \in \mathbb{F}_2^L$, the number of solutions $\tilde{\mathbf{t}}$ to $\mathbf{c} = H {\tilde{\mathbf{t}}}$ is $2^{2B-2L}$ and, if ${\tilde{\mathbf{t}}}$ is drawn uniformly at random from $\mathbb{F}_2^{2B}$,
\begin{align*}
    \Pr(H {\tilde{\mathbf{t}}} = {\mathbf{c}}) = \frac{2^{2B-2L}}{2^{2B}} = 2^{-2L}.    
\end{align*}
This implies that the column vector $\begin{bmatrix}
x^L_i \\ x^L_j
\end{bmatrix}$ is chosen uniformly at random from $\mathbb{F}_2^{2L}$.
This in turn implies that the entries of $x^L_i$ and $x^L_j$ are all mutually independent  i.i.d.~$\Bernoulli(1/2)$ random variables.

Given this fact, the event that $y^L_i$ and $y^L_j$ are consistent is the intersection of $L$ independent events
\begin{align*}
    \{ x^L_i[\ell] = x^L_j[\ell] \text{ or } x^L_i[\ell] = \varepsilon \text{ or } x^L_j[\ell] = \varepsilon \},    
\end{align*}
for $\ell=1,\ldots,L$.
Each of these events happens with probability
$1-\tfrac12(1-p)^2$,
implying that $x^L_i$ and $x^L_j$ are consistent with probability 
\begin{align*}
    (1-\tfrac12(1-p)^2)^L = 2^{-\gamma \log M} = M^{-\gamma}.
\end{align*}
Finally, we notice that the expected number of output sequences is $M E[N_1]$, and the expected number of pairs 
of output strings 
is at most $M^2 E[N_1]^2$.
Hence, the expected number of incorrect edges satisfies
\begin{align*}
    E[Z] \leq M^2 E[N_1]^2 M^{-\gamma} = E[N_1]^2 M^{2-\gamma}.    
\end{align*}

Finally, using Markov's inequality, we have that 
\begin{align*}
    \Pr\left(Z > M^{2-\gamma + \epsilon}\right)
\leq \frac{E[Z]}{M^{2-\gamma + \epsilon}} &\leq 
\frac{E[N_1]^2 M^{2-\gamma}}{M^{2-\gamma + \epsilon}}\\
&= E[N_1]^2 M^{-\epsilon},    
\end{align*}
which tends to $0$ as $M \to \infty$ for any $\epsilon > 0$.
\end{proof}

{\footnotesize
\bibliographystyle{IEEEbib}
\bibliography{refs}
}


\end{document}